\documentclass{article}

\usepackage{amsmath, amssymb, amsthm, mathrsfs} 
\usepackage{accents}
\usepackage{algorithm, algpseudocode}
\usepackage{graphics, graphicx} 
\usepackage[margin=1in]{geometry} 
\usepackage{theoremref}
\usepackage[title]{appendix}

\usepackage{fancyhdr} 
\usepackage{url} 
\usepackage[normalem]{ulem}
\usepackage[dvipsnames]{xcolor} 

\usepackage{setspace}

\usepackage{accents}
\usepackage{tikz}
\usetikzlibrary{positioning}
\usepackage{caption}

\newcommand{\sets}{\ensuremath{\mathcal{S}}}
\newcommand{\one}{\ensuremath{\mathbf{1}}}

\newcommand{\prob}{\ensuremath{\mathbf{P}}}
\newcommand{\expec}{\ensuremath{\mathbf{E}}}
\newcommand{\ind}{\ensuremath{\mathbf{I}}}
\newcommand{\reals}{\ensuremath{\mathbb{R}}}
\newcommand{\naturals}{\ensuremath{\mathbb{N}}}
\newcommand{\defeq}{\ensuremath{\triangleq}}
\newcommand{\sP}{\ensuremath{\mathsf{P}}}
\newcommand{\sQ}{\ensuremath{\mathsf{Q}}}

\newcommand{\ubar}[1]{\underaccent{\bar}{#1}}
\newcommand{\ubarc}{\ubar{C}}
\newcommand{\barc}{\bar{C}}
\newcommand{\tildec}{\tilde{C}}

\newcommand{\tilden}{\tilde{n}}

\newcommand{\vhat}{\hat{V}}

\newtheorem{theorem}{Theorem}[section]

\newtheorem{lemma}{Lemma}[section]

\newtheorem{definition}{Definition}

\newcommand{\maskforsubmission}[1]{ }

\title{Mean Field Equilibria for Competitive Exploration in Resource Sharing Settings}

\author
{Pu Yang\thanks{py75@cornell.edu},  Krishnamurthy Iyer\thanks{kriyer@cornell.edu}, and Peter Frazier\thanks{pf98@cornell.edu}\\
\\
\normalsize{School of Operations Research and Information Engineering, Cornell University}}

\date{}

\begin{document} 


\maketitle 

\begin{abstract}
We consider a model of nomadic agents exploring and competing for time-varying
location-specific resources, arising in crowdsourced transportation services,
online communities, and in traditional location-based economic activity.
This model comprises a group of agents, and a set of locations each endowed
with a dynamic stochastic resource process. 
Each agent derives a periodic reward determined by the overall resource level
at her location,
and the number of other agents there.
Each agent is strategic and free to move between locations, and at each time
decides whether to stay at the same node or
switch to another one. We study the equilibrium behavior of the agents as a
function of dynamics of the stochastic resource process and the nature of the
externality each agent imposes on others at the same location. In the
asymptotic limit with the number of agents and locations increasing
proportionally, we show that an equilibrium exists and has a threshold
structure, where each agent decides to switch to a different location based
only on their current location's resource level and the number of other agents
at that location. This result provides insight into how system structure
affects the agents' collective ability to explore their domain to find and
effectively utilize resource-rich areas. It also allows assessing the impact of
changing the reward structure through penalties or subsidies.
\end{abstract}

\section{Introduction}
\label{sec:intro}

We consider a model of nomadic agents exploring and competing for
time-varying stochastic location-specific resources.  Such multi-agent
systems arise in crowdsourced transportation services like Uber and Lyft
where drivers position themselves to be close to demand; in online
communities like Twitch and Reddit where webizens choose in which
subcommunities to participate; and in location-specific activity in
the traditional economy, such as food trucks choosing where to position
themselves, fisherman choosing where to fish, and pastoralists
choosing where to graze their animals.  In each of these examples,
overall social welfare is determined both by agents' willingness to
explore their domain to find and exploit resource-rich locations, the
level of antagonism or synergy inherent to having multiple agents at
the same location, and the equilibrium distribution of agents across
locations that these factors induce.

The model we study comprises a set of locations and a group of agents.
Each location has a resource level that varies randomly with time.
Each agent periodically derives resource from the location at which
she currently resides, whose amount is determined by the number of
other agents currently residing there, and the location's current
overall resource level.  Based on these quantities, the agent then
decides whether to stay at the same location or switch to another.
The agents are fully strategic and seek to maximize their total
rewards over their lifetime.

We study the equilibrium behavior of the agents in this system as a
function of dynamics of the spatio-temporal resource process and the
level of synergy or antagonism in the agents' sharing of resources.
We analyze the system under the limit where the number of agents and
locations both increase proportionally, using the methodology of a
mean field equilibrium. We show that an equilibrium exists, and the
agents' optimal strategy has a simple threshold structure, in which it
is optimal to leave a location when the number of other agents exceeds
a threshold that depends on the resource level at that location.  In
the limit as the system grows large, this induces a joint probability
distribution over the number of agents and level of resource at each
location.

Our results allow us to obtain economic insights into how the nature
of the externality agents impose on others at the same location
affects the exploration of the locations for resources, and
consequently the overall welfare of the economy. In particular, our
methodology allows us to analyze settings where the overall reward at
a location either increases or decreases with the number of agents at
the location, and how these two settings affect the equilibrium
exploration. Furthermore, our methodology allows us to evaluate
engineering interventions, such as providing subsidies to or imposing
costs on agents to promote or discourage exploration to improve
welfare.

\paragraph{Examples}

The model we study is a simplified version of systems appearing in
real-world settings.  It arises in the shared economy, in crowdsourced
transportation services such as Uber and Lyft, in which drivers choose
neighborhoods, and then earn money (reward) based on the number of
riders requesting service within that neighborhood (the overall
resource level), and the number of other drivers working there.  This
overall resource level varies stochastically over time in a
neighborhood-specific way as demand rises and falls, and the resource
derived by a driver decreases with the number of other drivers working
in her neighborhood.

This model also arises in the internet economy, in online communities
such as Reddit and Twitch, in which participants choose
sub-communities, and then derive enjoyment depending both on some
underlying but transitory societal interest in the sub-community's
topic of focus (the overall resource) and the number of other
participants in the sub-community.  When the number of other
participants is too small, lack of social interaction prevents
enjoyment; when the number of other participants is too large,
crowding diminishes the sense of community.

This model also arises in the traditional economy, for example in food
trucks deciding in which neighborhoods to locate, in pastoralists
deciding where to graze their livestock, and in fishermen deciding
where to fish.  In these examples, the level of resource derived by
each agent from their location (whether profit from hungry passers by;
or food for livestock provided by the range-land; or profit from the
catch) depends both on the number of other agents at the location, and
on the location's overall and stochastically varying resource level.

This model even arises among scientific researchers, who must choose a
research area in which to work, and derive value from this choice
based both on the underlying level of societal interest and funding in
their chosen area, and in the number of other researchers working in
it. As with online communities, the number of other researchers should
be neither too large nor too small too maximize the value derived.

\paragraph{Related Work}

Our paper adds to the growing literature on mean field equilibrium
\cite{sachin_2010,huang_2007,jovanovic_1988,lasry_2007,weintraub_2008},
that studies complex systems under a large system limit and obtains
insights about agent behavior that are hard to obtain from analyzing
finite models. The main insight behind this line of literature, that
in the large system limit, agents' behavior are characterized by their
(private) state and an aggregate distribution of the rest of system,
has been used to study settings that include industry dynamics and
oligopoly models \cite{hopenhayn_1992,weintraub_2008,weintraub_2010},
repeated dynamics auctions \cite{balseiro2014, iyer2014}, online labor
markets \cite{arnosti}, and queueing systems \cite{xu2013supermarket}.

Our model can be seen as an extension of the Kolkata Paise Restaurant
Problem \cite{chakrabarti2009kolkata}.  In this game, each agent
chooses (simultaneously) a restaurant to visit, and earns a reward
that depends both on the restaurant's rank, which is common across
agents, and the number of other agents at that restaurant.  This
reward is inversely proportional to the number of agents visiting the
restaurant.

The Kolkata Paise Restaurant Problem is itself a generalization of the
El Farol bar problem
\cite{arthur1994inductive,chakrabarti2007kolkata}.  The Kolkata Paise
Restaurant Problem is studied both in the one-shot and repeated
settings, with results on the limiting behavior of myopic
\cite{chakrabarti2009kolkata} and other strategies
\cite{ghosh2010statistics}, although we are not aware of existing
results on mean-field equilibria in this model.  Our model is both
more general, in that we allow general reward functions and allow
location's resource to vary stochastically, and more specific, in that
our locations are homogeneous.  Our model also differs in that our
agents' decisions are made asynchronously.

Our model is also related to congestion games
\cite{nisan2007algorithmic,rosenthal73}, in which agents choose paths
on which to travel, and then incur costs that depend on the number of
other agents that have chosen the same path.  One may view paths as
being synonymous with locations in our model, and observe that in both
cases the utility/cost derived from a path/location depends on the
number of other agents using that path, or portion thereof. The main
difference between our model and congestion games is the stochastic
time-varying nature of our overall level of resource (making our model
more complex), and the lack of interaction between locations
contrasting with the interaction between paths (making our model
simpler).

Our model has within it an exploration vs. exploitation tradeoff, in
which an agent faces the decision of whether to stay at his current
node, exploiting its resource and obtaining a known reward, or to
leave and go to another randomly chosen location with unknown reward.
Visiting this new location provides information upon which future
decisions may be based. Similar tradeoffs between exploration and
exploitation appear widely, and have been studied extensively in the
single-agent setting
\cite{AuCeFi02,Gi89,Ka93,PoRy12}. Exploration and learning in multi-agent
settings has been considered
by \cite{FrazierKempeKleinbergKleinberg2014,Lobel2015}.

\section{Model}
\label{sec:model}

We consider a setting with $N$ agents, each situated at each time
$t \geq 0$ in one of $K$ locations. Each location $k$ has a stochastic
dynamic resource process, denoted by $\{Z_t^{(k)}: t \geq 0\}$, that
determines the reward obtained by each agent at that location, as we
describe below. We assume that each process $\{Z_t^{(k)}: t \geq 0\}$
is a finite state continuous time Markov chain, that is distributed
identically and independently from the rest of the system. For the
purpose of analysis, we assume that 
each $Z_t^{(k)}$ takes
values in $\{0, 1\}$, with holding time at state $z \in \{0,1\}$
distributed as $\textsf{Exp}(\mu_{z,1-z})$ for some fixed
$\mu_{z,1-z} > 0$.

Agents may switch between locations to explore for
resources. Formally, each agent $i$ has an associated independent
Poisson process $X^i_t$ with rate $\lambda > 0$, at whose jump times
$\{T^i_\ell : \ell \geq 0\}$ the agent makes the decision to either
stay in her current location or switch instantaneously to a different
location that is chosen uniformly at random. Let $k^i_t$ denote the
location of agent $i$ at time $t \geq 0$ and let $N^{(k)}_t$ denote the
number of agents at location $k$ at time $t$.

The agents in the model are short-lived, and at each jump time
$T^i_\ell$, the agent $i$ departs the system with probability
$1-\gamma > 0$. Thus, each agent $i$ lives in the system for a random
time $\tau^i$ which is distributed according to
$\textsf{Exp}(\lambda (1-\gamma))$. For each agent $i$ that leaves the
system, a new agent (with the same label $i$) arrives at a location
chosen uniformly at random. We make this modeling assumption to ensure
that the number of agents in the system is always positive; this
assumption can be relaxed to allow for random arrivals and departures,
with the arrival rate equal to the departure rate.

We now describe the decision problem faced by an agent $i$ in more
detail. At each jump time $t =T^i_\ell$, an agent $i$ at location
$k = k^i_t$ receives a reward $R^i_t = F(Z^{(k)}_t,N^{(k)}_t)$ that
depends on the state of the resource process $Z_t^{(k)}$ and the overall
number of agents $N_t^{(k)}$ at the location $k$. In the following, we
assume that the function $F$ governing the reward at each location is
given by $F(z, n) = zf(n)$ for some function
$f:\naturals \to \reals^+$ that is non-increasing in $n$, with
$\lim_{n\rightarrow +\infty} f(n) =0$.  Essentially, this implies that
the reward at a location at any time is zero if the resource process
is in state $0$, and it is equal to $f(n)$ if the resource process is
in state $1$, if there are $n$ agents at that location. Furthermore,
this reward $f(n)$ decreases to zero as the number of agents at a
location increases. 
Given this
setting, each agent $i$ at any time prefers to be at a location with
resource process state equal to $1$, and where the number of other
agents is small.

Within this setting, we will focus on three cases: (1) for
each $n \in \naturals$, we have $nf(n) = 1$. In this 
case the (unit) resource, if available at a location, is shared
equally among the agents at that location; (2) $nf(n)$ is
non-decreasing in $n$. In this case adding agents to a location 
increases the total reward earned, either through synergy, 
or because a small number of agents cannot fully utilize a location's resource;
and (3) $nf(n)$ is non-increasing in the number
of agents $n$ at that location. In this case antagonism or overutilization
causes the total reward earned to decrease as agents are added.

Next, we discuss the information each agent $i$ has while making their
decision to stay or switch. We assume that an agent $i$ has access to
the states of the resource process $Z^{(k)}_t$ and the number of
agents $N^{(k)}_t$ of a location $k$ during the time she is present at
the location $k$, i.e., when $k^i_t = k$. We further assume that
agents have perfect recall, and hence, at a jump time $t = T^i_\ell$,
each agent $i$ bases her decision to switch or stay on the entire
history $h^i_t$ she has observed until time $t$, namely the resource
process states and the number of agents at each location she has
visited during the time period she visited that location: 
\begin{align*}
h^i_t = \left\lbrace (Z_s^{(k^i_s)}, N_s^{(k^i_s)}) : s \leq t\right\rbrace.
\end{align*}
Thus, a strategy $\xi^i$ for an agent $i$, specifies a (mixed) action
between stay and switch at each jump time $t=T^i_\ell$ of her
associated Poisson process $X^i_t$, based on her history $h^i_t$.

Given this informational assumption, each agent $i$ seeks to maximize
the total expected reward accrued over her lifetime, given by
\begin{align*}
  \expec\left[ \sum_{\ell = 0}^\infty R^i_{T^i_\ell} \ind\{\tau^i \geq T^i_\ell\} \right].
\end{align*}
Observe that since the agent departs the system with probability
$(1-\gamma)$ at each jump time independently, the total expected
reward can be equivalently written as
\begin{align*}
  \expec\left[ \sum_{\ell = 0}^\infty \gamma^\ell R^i_{T^i_\ell}\right].
\end{align*}
Thus, each agent's decision problem is equivalent to
maximizing her total discounted expected reward assuming she
persists in the system.

Since the reward at any location is determined by the number of agents
at that location, each agent's decision to stay in her current location
or to switch to a new one depends on all the other agents'
behavior. Consequently, the interaction among the agents is a dynamic
game, and analyzing the agents' behavior requires an
equilibrium analysis.

The standard equilibrium concept to analyze the induced dynamic game
is a perfect Bayesian equilibrium (PBE). A PBE consists of a strategy
$\xi^i$ and a belief system $\mu^i$ for each player $i$. A belief
system $\mu^i$ for agent $i$ specifies a belief $\mu^i(h^i_t)$ after
any history $h^i_t$ over all aspects of the system that she is
uncertain of and that influence her expected payoff. A PBE then
requires two conditions to hold: (1) each agent $i$'s strategy $\xi^i$
is a best response after any history $h^i_t$, given their belief
system and given all other agents' strategies; and (2) each agent
$i$'s beliefs $\mu^i(h^i_t)$ are updated via Bayes' rule whenever
possible (see \cite{tirole,fudenbergT91} for more details).

Observe that a PBE supposes a complex model of agent behavior. It
requires each player $i$ to keep track of her entire history, and
maintain complex beliefs about the rest of the system. While this may
be plausible in small settings, this behavioral model seems
implausible for large systems. On the contrary, in such settings, it
is more plausible that each agent would base her decision to stay or
switch solely on the current state of the location she is in ---
specifically on its level of resource, and the number of other agents
there--- and on the aggregate features of the entire system.
Moreover, we expect that if an agent were to base her decision only on
this information, then she would pursue a ``threshold'' strategy: she
would stay in her current location if the number of agents at that
location is low, and switch to a different location if that number is
high, with the threshold used depending on that location's level of
resource.

Below, we seek to uncover this intuitive behavioral model as an equilibrium
in large systems by letting the number of
agents and the number of location both increase proportionally to
infinity, and studying the limiting infinite system.

\section{Limiting infinite system}
\label{sec:limiting}

In this section, we consider an infinite system that is obtained as
the limit of the finite system as the number of location $K$ and the
number of agents $N$ both tend to infinity, with $N = \beta K$, for
some fixed $\beta >0$. In the limiting system, there are infinite
number of locations and agents, with the expected number of agents per
location fixed at $\beta > 0$. In such a limit, given certain
consistency conditions that bind the mean dynamics of all the
locations, the dynamics of each location essentially decouples from
the rest of the system.  Under such a decoupling, instead of focusing
on the entire limiting system, it suffices to focus on the dynamics of
a single location, as well as the empirical distribution of the states
of all the locations. We begin with the description of the dynamics of
a single location in such an infinite system.

\subsection{Location dynamics}

To analyze the agents' behavior in the infinite system, we fix a
location $k$ and focus on the decision problem faced by an agent $i$
at location $k$ about when to switch to a different location. Let
$X^i_t$ denote the Poisson process with rate $\lambda >0$ associated
with agent $i$, with jump times $T^i_\ell$ for $l\geq 0$.  As before,
let $Z^{(k)}_t$ denote the state of the resource process at location
$k$ and let $N^{(k)}_t$ denote the number of agents at location $k$ at
time $t$. We assume that agents arrive at location $k$ according to a
Poisson arrival process with rate $\kappa>0$. Note that these arriving
include new agents arriving to the system (following a departure), as
well as existing agents who have chosen to switch from their current
location.

Inspired by the discussion at the end of the preceding section, we
focus on a family of threshold strategies for the agents. A threshold
strategy is characterized by a pair $(n_0, n_1) \in [0, +\infty)^2$.
In a threshold strategy $(n_0, n_1)$, an agent at a location with
resource level $z \in \{0,1\}$ chooses to stay at her current location
if the number of agents is strictly below $\lfloor n_z \rfloor$;
chooses to switch her location if the number of agents at her current
location is strictly above $n_z$; and stays with probability
$n_z - \lfloor n_z \rfloor$ and switches with the remaining
probability if the number of agents is equal to $\lfloor
n_z\rfloor$. Our eventual goal is to show that there exists an
equilibrium for agents' behavior where all agents follow (the same)
threshold strategy. For now, we assume that all agents except agent
$i$ adopt a threshold strategy $(n_0,n_1)$, and seek an optimal
strategy over the class of all history-dependent strategies (not just
the class of threshold strategies) for agent $i$.

Note that given the arrival rate $\kappa>0$, and the threshold policy
$(n_0, n_1)$, the process $(Z^{(k)}_t,N^{(k)}_t)$ evolves as a
continuous time Markov chain on the state space
$\sets = \{0,1\} \times \naturals$ with the following transition rate
matrix: for each $z \in \{0,1\}$, and for all $n \in \naturals$, we
have
\begin{align*}
  \sQ((z,n) \to (1-z,n)) &= \mu_{z,1-z},\\
  \sQ((z,n) \to (z,n+1)) &= \kappa,\\  
  \sQ((z,n) \to (z,n-1))  &= \lambda (n-1) \left(1- \gamma\ + \right. \\ 
\gamma \left(\one\{n > n_z\} \right. & + \left.\left. (n+1-n_z)\one\{n = \lfloor n_z \rfloor\}\right)\right).  
\end{align*}
Here, the first equation represents the transitions in
$Z^{(k)}_t$, which is an independent Markov chain on $\{0,1\}$ with
holding times $\mu_{01}$ and $\mu_{10}$. The second equation follows
from the assumption that agents arrive at location $k$ according to a
Poisson process with rate $\kappa>0$. The third equation represents a
transition where an agent at location $k$ leaves. This transition can
occur in two ways: first, the agent could leave the system with
probability $1-\gamma$; second, the agent could survive, with
probability $\gamma$, but choose to switch to a different location,
which happens with probability $1$ if $n > n_z$, with probability
$(n+1-n_z)$ if $n = \lfloor n_z \rfloor$, and zero otherwise. Since
there are $(n-1)$ other agents that make this decision to stay or
switch at rate $\lambda$, these transitions occur at rate
$\lambda(n-1)$. We denote this continuous time Markov chain describing
the dynamics of a single location, where all agents adopt the
threshold policy $(n_0, n_1)$ and the rate of arrival of agents is
$\kappa$, by $\mathcal{MC}(n_0, n_1, \kappa)$.

\subsection{Agent's decision problem}

We are now ready to describe the decision problem faced by the agent
$i$ regarding when to switch from her current location. At each jump
time $t = T^i_\ell$ of $X^i_t$, the agent $i$ receives an immediate
payoff of $Z_t^{(k)}f(N^{(k)}_t)$ and may leave the system with
probability $1-\gamma$. If she does not leave the system, then she has
to decide between two actions ``stay'' or ``switch''. On choosing
``stay'' continues until the next jump time $T^i_{\ell+1}$; on
choosing ``switch'', the decision problem terminates with an immediate
payoff of $C >0$, that does not depend on the state of the location
$k$.

Before proceeding, we provide a brief interpretation of the
termination payoff $C$. Observe that in a finite system, an agent on
switching from a location, moves on to a different location that is
chosen uniformly at random, and continues to accrue payoffs until she
leaves the system. This suggests that one may interpret the termination
payoff $C$ as capturing the notion of a continuation payoff on
switching in the finite system in the context of the limiting infinite
system. Subsequently, we impose conditions on our equilibrium notion
that ensure that indeed $C$ denotes the continuation payoff in the
infinite system.

\emergencystretch=.5em
Given these payoffs and actions for agent $i$, it follows that the
decision problem facing agent $i$ is an optimal stopping problem,
which we denote by $\mathcal{OS}(n_0, n_1, \kappa, C)$. We next
specify the dynamic programming formulation of 
$\mathcal{OS}(n_0, n_1, \kappa, C)$.

Note that in the decision problem, when the
Markov chain $\mathcal{MC}(n_0, n_1, \kappa)$ is in state $(z,n)$,
events occur at rate $\kappa + \mu_{z,1-z} + \lambda n$: with rate
$\kappa$ a new agent arrives, with rate $\mu_{z, 1-z}$ the resource
level changes, with rate $\lambda(n-1)$ one of the other agents either
leaves the system or survives and makes the decision to stay or
switch, and finally with rate $\lambda$, agent $i$ arrives at a jump
time to make a decision herself. Thus, we define the following transition probabilities for the state transition:
\begin{align}
  \label{eq:transition-probabilities}
  \sP_{dec}(z,n) &= \frac{\lambda}{ n\lambda + \mu_{z,1-z}  + \kappa}, \notag \\
  \sP_{exit}(z,n) &= \frac{(n-1)\lambda(1-\gamma)}{n\lambda + \mu_{z,1-z}  + \kappa},\notag \\
  \sP_{sur}(z,n) &= \frac{(n-1)\lambda\gamma}{ n\lambda + \mu_{z,1-z}  + \kappa},\notag \\
  \sP_{res}(z,n) &= \frac{\mu_{z,1-z}}{n\lambda + \mu_{z,1-z}  + \kappa}, \notag\\
  \sP_{arr}(z,n) &= \frac{\kappa}{n\lambda + \mu_{z,1-z}  + \kappa}.
\end{align}
Here, $\sP_{dec}(z,n)$ denotes the probability that the next event
corresponds to agent $i$'s decision epoch, $\sP_{exit}$ denotes the
probability the next event corresponds to one of the other agents
exiting the system, $\sP_{sur}$ denotes the probability the next event
corresponds to one of the other agents persisting in the system,
$\sP_{res}$ denotes the probability the next event corresponds to
change in the resource level, and finally, $P_{arr}$ denotes the
probability that the next event corresponds to a new arrival.

Given $(Z^{(k)}_t, N^{(k)}_t) = (z,n)$ with $T^i_\ell = t$, let
$V(z,n)$ denote the optimal expected total reward of agent $i$ just
after her associated Poisson process $X^i_t$ has undergone a jump, but
prior to her making a decision or receiving any reward. Similarly, let
$\hat{V}(z,n)$ denote the optimal expected total reward of agent $i$
after a jump time, conditional on the decision problem not terminating
either due to the agent leaving the system or choosing to switch to a
different location. Then, we have the following Bellman equation:
\begin{align}
  \label{eq:bellman}
  &V(z,n) = zf(n) + \gamma \max\{ \vhat(z,n), C\}, \notag \\
  &\vhat(z,n) = \sP_{dec}(z,n) V(z,n) + \sP_{exit}(z,n)\vhat(z,n-1)  \notag \\
    & + \sP_{sur}(z,n) \left(\one{\{ n\!>\! n_z \}} + (n\!+\!1\!-\!n_z)\one{\{n\!=\!\lfloor n_z \rfloor \}}\right) \vhat(z, n\!-\!1) \notag \\
    & + \sP_{sur}(z,n) \left( \one{\{n\!<\!\lfloor n_z \rfloor \}} + (n_z\!-\!n) \one{\{ n\!=\!\lfloor n_z \rfloor \}}\right)\vhat(z, n) \notag \\
    & + \sP_{res}(z,n) \vhat(1\!-\!z, n) + \sP_{arr}(z,n) \vhat(z, n\!+\!1).
\end{align}
Here, the first equation follows from the fact that subsequent to a
jump time, the agent receives an immediate payoff equal to
$zf(n)$. Following this, she continues with survival probability
$\gamma$, and has to make a decision to stay, which gets her expected
payoff equal to $\hat{V}(z,n)$ or switch, which gets her an expected
payoff equal to $C$. The second equation relates $\hat{V}(z,n)$ to the
agent's expected payoff subsequent to various events that can occur at
the next transition. For a solution $V$ and $\hat{V}$ to the Bellman's
equation, an optimal strategy $\xi^i$ for the agent $i$ requires agent
$i$ to stay if $\hat{V}(z,n) > C$, and to switch if $\hat{V}(z,n) < C$
(any mixed action is optimal if $\hat{V}(z,n) = C$). Let
$\mathcal{OPT}(n_0,n_1,\kappa,C)$ denote the set of all optimal
strategies (not necessarily threshold strategies) for the agent's
decision problem $\mathcal{OS}(n_0,n_1,\kappa, C)$.

\subsection{Mean field equilibrium}

Given the Markov chain $\mathcal{MC}(n_0, n_1, \kappa)$ and an agent's
decision problem $\mathcal{OS}(n_0, n_1, \kappa, C)$ , we are now
ready to state the equilibrium conditions on the limiting
system. First, in the infinite system, we require the agents'
strategies to be in equilibrium. Since
$\mathcal{MC}(n_0, n_1, \kappa)$ describes the dynamics of a location
where all agents other than agent $i$ use the threshold policy
$(n_0, n_1)$, for equilibrium we must impose the condition that
$(n_0, n_1)$ is an optimal strategy for the agent's decision
problem. This leads to the following condition:
\begin{align}
\label{eq:best-response}
(n_0, n_1) \in \mathcal{OPT}(n_0,n_1,\kappa, C).
\end{align}

If all agents at location $k$, including agent $i$, follow the
threshold policy $(n_0, n_1)$, then the transitions in
$(Z^{(k)}_t, N^{(k)}_t) = (z,n)$ follow a Markov chain with transition
rate matrix $\overline{\sQ}$ that is equal to $\sQ$ except for the
transition $\overline{\sQ}((z,n) \to (z,n-1))$ which is equal to
$\overline{\sQ}((z,n) \to (z,n-1)) = \lambda n (1- \gamma + \gamma
(\one\{n > n_z\} + (n+1-n_z)\one\{n = \lfloor n_z \rfloor\} ))$.
This is because the arrival and the changes in the resource level
occur at the same rate, but now any one of the agents at location $k$
might choose to leave the location, as opposed to any one of the
agents other than agent $i$ as defined in $\sQ$. Denote this Markov
chain by $\overline{\mathcal{MC}}(n_0, n_1, \kappa)$ and let $\pi$
denote an invariant distribution of this chain:
\begin{align}
\label{eq:steady-state}
\pi^{\mathsf{T}}\overline{\sQ} = 0.
\end{align}
In a large system, a natural requirement to impose is that the
invariant distribution of a single location $k$ equals the steady
state empirical distribution of the resource level and the number of
agents across all locations. Requiring this condition to hold leads to
two consequences. First, because in the infinite system the ``agent
density'', i.e., mean number of agent across all locations, is equal
to $\beta$, this implies that the expected number of agents at
location $k$ must equal $\beta$:
\begin{align}
\label{eq:expectation-steady}
\sum_{(z,n)\in\sets} n\pi(z,n) = \beta.
\end{align}
Note that this equation imposes a restriction on the arrival rate
$\kappa$ of the Markov chain $\mathcal{MC}(n_0, n_1, \kappa)$. In
particular, it requires the arrival rate to be such that in steady
state the expected number of agents at each location is equal to
$\beta$.

The second condition imposes a restriction on the immediate
termination reward on switching $C$. Recall that we interpret $C$ as
modeling the optimal continuation payoff on switching in the finite
system. Since the empirical distribution of the states of other
locations is given by $\pi$, the optimal expected reward an agent can
obtain on switching is given by
$\sum_{(z,n)\in\sets} \pi(z,n) \vhat(z,n+1)$. This is because, after
the agent moves to a location in state $(z,n)$, which happens with
probability $\pi(z,n)$, the number of agents at that location becomes
$n+1$, and the expected payoff to that agent is $\hat{V}(z,n+1)$. We
require that this quantity equals the immediate reward $C$:
\begin{align}\label{eq:leave-payoff-consistency}
    C = \sum_{(z,n)\in\sets} \pi(z,n) \vhat(z,n+1).
\end{align}

Given these consistency conditions, we are now ready to define a mean
field equilibrium for the infinite system:
\begin{definition}[Mean Field Equilibrium] A mean field equilibrium is
  characterized by a threshold strategy $(n_0, n_1)$, an arrival rate
  $\kappa>0$, an distribution $\pi$ over $\sets$, and an immediate
  reward $C>0$, such that the set of equations \eqref{eq:best-response},\eqref{eq:steady-state},\eqref{eq:expectation-steady}, and \eqref{eq:leave-payoff-consistency} hold. 
\end{definition}

Note that as opposed to a PBE, a mean field equilibrium adopts a
fairly natural and simple model of agent behavior, where each agent
needs to keep track only of current state and the number of agents at
the location she is in, along with the immediate payoff for switching.

\section{Main results}
\label{sec:main-thm}

Having defined the equilibrium concept, we now consider the problem of
existence of a mean field equilibrium in the infinite system. 

We begin with the following lemma that shows that for any level of
resource at a location, the value function $\hat{V}(z,n)$ is
non-increasing with the number of agents at that location. 
The proof may be found in Appendix~\ref{ap:few-proofs}.
\begin{lemma}
  \label{lem:value-decreasing}
  For each $z\in \{0,1\}$, the value function $\vhat(z,n)$ is
  non-increasing in $n$.
\end{lemma}
Using this lemma, it is straightforward to show that for any level of
resource $z \in \{0,1\}$, if it is optimal to switch when the number
of agents at the location is $n$, then it is still optimal to switch
when the number of agents is greater than $n$. From this, we obtain
the first result of this section.
\begin{theorem}
\label{thm:threshold}
For any $(n_0, n_1) \in \reals^2_+$, $\kappa > 0$ and $C >0$, there
always exists an optimal strategy with a threshold structure for the
decision problem $\mathcal{OS}(n_0, n_1, \kappa, C)$.
\end{theorem} 
The preceding theorem suggests that set of threshold strategies is
closed under best-responses: if all agents other than agent $i$ adopt
the same threshold strategy, then it is optimal for agent $i$ to also
follow a threshold strategy. Thus, it suffices to focus on the set of
optimal threshold strategies for the agent decision problem, which we
denote by $\mathcal{T}(n_0, n_1, \kappa, C)$. Note that
$\mathcal{T}(n_0, n_1, \kappa, C)$ can be characterized as a subset of
$\reals^2_+$ corresponding to the values of the thresholds of the
optimal threshold strategies. In Lemma~\ref{lem:thresholds-convex} in the
Appendix~\ref{ap:few-proofs}, we show that $\mathcal{T}(n_0, n_1, \kappa, C)$
is a convex set.

Building on this result, we obtain the main theorem of our paper.
\begin{theorem}
  \label{thm:existence} For any $\lambda>0$, $\beta > 0$ and
  $\mu_{01}, \mu_{10} > 0$, there exists a mean field equilibrium for
  the infinite system.
\end{theorem}
We emphasize that the existence of a mean field equilibrium is
obtained under very general conditions, requiring only that the reward
function $f(n)$ is non-increasing in the number of agents $n$ at a
location, and converging to zero as $n$ tends to infinity. Our proof
technique can be extended to cases where the resource level
$Z^{(k)}_t$ at any location can take values in any finite subset of
$\reals_+$.

The full proof of Theorem~\ref{thm:existence} is technical and is omitted due
to space constraints.  Instead, in the next section, we sketch the main ideas
behind the proof and provide a brief outline.  Selected portions of the proof
may be found in Appendix~\ref{ap:few-proofs}, and a full proof will appear
later in a longer version of the paper.

\section{Proof outline}
\label{sec:proof-outline}

The proof of Theorem~\ref{thm:existence} follows by applying
Kakutani's fixed point theorem on carefully defined map $\mathscr{T}$,
whose fixed points correspond to the mean field equilibria of the
infinite system. To define the map $\mathscr{T}$ requires a number of
intermediate steps, which we outline below.

The first step of the proof involves showing that given any
$(n_0, n_1) \in \reals^2_+$ and $C>0$, there exists a unique $\kappa$
and distribution $\pi$ such that $\pi$ is the unique invariant
distribution of the Markov chain
$\overline{\mathcal{MC}}(n_0, n_1, \kappa)$ (i.e., $\pi$ satisfies
equation \eqref{eq:steady-state}), and for which condition
\eqref{eq:expectation-steady} holds. This step itself involves first
showing that for any $\kappa>0$, the Markov chain
$\overline{\mathcal{MC}}(n_0,n_1,\kappa)$ is irreducible and positive
recurrent, and hence has a unique stationary distribution
$\pi_\kappa$. To show this, we use coupling arguments to bound the
Markov chain $\overline{\mathcal{MC}}(n_0,n_1, \kappa)$ between two $M/M/\infty$
queues. Then, we show that the quantity
$\sum_{(z,n)\in\sets} n\pi_\kappa(z,n)$ is strictly increasing and
continuous over an interval of values of $\kappa$, which suffices to
show that there exists a value of $\kappa$ satisfying
\eqref{eq:expectation-steady}.

We then compute the value function $\hat{V}$ satisfying the Bellman's
equation \eqref{eq:bellman} for the decision problem
$\mathcal{OS}(n_0,n_1, \kappa, C)$, where $\kappa$ is the value of the
arrival rate obtained in the first step. Using this value function, we
identify the set $\mathcal{T}(n_0, n_1, \kappa, C)$ of optimal
threshold strategies.

Finally, using the value function $\hat{V}$ and the invariant
distribution $\pi$ from the first step, we compute the total expected
payoff of an agent subsequent to switching to a different location
chosen uniformly at random, defined as
$\tilde{C}\defeq \sum_{(z,n)\in\sets} \pi(z,n) \vhat(z,n+1)$.

The map $\mathscr{T}$ is then defined as follows: for each
$(n_0, n_1) \in \reals^2_+$ and $C>0$, we define
$\mathscr{T}(n_0, n_1, C) = \mathcal{T}(n_0, n_1, \kappa, C) \times
\{\tilde{C}\}$. We depict the map pictorially in Fig.~\ref{fig:map-t}.

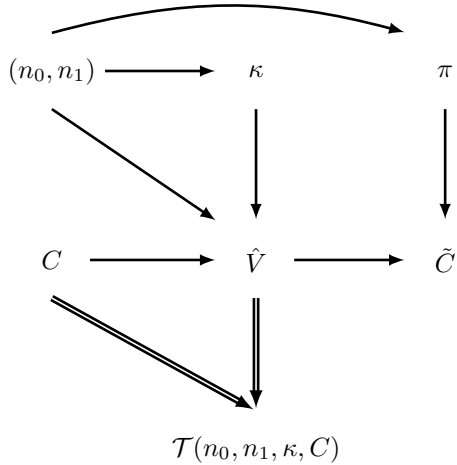
\begin{figure}
  \centering
  \begin{tikzpicture}[plainnode/.style={draw=none, fill=none, minimum size = 10mm}, singlearrow/.style={thick,->,shorten >=1pt, line width = 1pt, >=latex}, doublearrow/.style={thick, ->, shorten >= 1pt, line width = 1pt, >=latex, double}]
    \node[plainnode]      (originalthreshold)                        {$(n_0, n_1)$};
    \node[plainnode]      (kappa)             [right=1.5cm of originalthreshold] {$\kappa$};
    \node[plainnode]      (pi)                [right=1.5cm of kappa] {$\pi$};
    \node[plainnode]      (C)                 [below=1.5cm of originalthreshold] {$C$};
    \node[plainnode]      (vhat)              [below=1.5cm of kappa] {$\vhat$};
    \node[plainnode]      (ctilde)            [below=1.5cm of pi] {$\tildec$};
    \node[plainnode]      (newthreshold)      [below=1.5cm of vhat] {$\mathcal{T}(n_0, n_1, \kappa,C)$};
     
    \draw[singlearrow] (originalthreshold.east) -- (kappa.west);
    \draw (originalthreshold.north) edge[out=15, in=165, singlearrow] (pi.north west);
    \draw[singlearrow] (C.east) -- (vhat.west);
    \draw[singlearrow] (vhat.east) -- (ctilde.west);
    \draw[singlearrow] (kappa.south) -- (vhat.north);
    \draw[singlearrow] (pi.south) -- (ctilde.north);
    \draw[singlearrow] (originalthreshold.south) -- (vhat.north west);
    \draw[doublearrow] (vhat.south) -- (newthreshold.north);
    \draw[doublearrow] (C.south) -- (newthreshold.north);
  \end{tikzpicture}
  \caption{Illustration of the correspondence $\mathscr{T}(n_0, n_1, C) = \mathcal{T}(n_0,n_1,\kappa,C) \times \{\tilde{C}\}$.}
\label{fig:map-t}
\end{figure}

By definition, any fixed point $(n_0,n_1, C)$ of the map $\mathscr{T}$
must satisfy $(n_0, n_1) \in \mathcal{T}(n_0, n_1, \kappa,C)$, and
$C = \tilde{C}$. This implies that $(n_0, n_1)$ is an optimal
threshold strategy for the decision problem
$\mathcal{OS}(n_0, n_1, \kappa, C)$, and hence is also an optimal
strategy and equation \eqref{eq:best-response} holds. Recall that the
arrival rate $\kappa$ and the invariant distribution $\pi$ satisfy
equations \eqref{eq:steady-state} and
\eqref{eq:expectation-steady}. Finally, $\tilde{C}=C$ implies that
equation \eqref{eq:leave-payoff-consistency} holds. From this we
conclude that $(n_0, n_1)$, $C$ and resulting $\kappa$ and the
invariant distribution $\pi$ together constitute a mean field
equilibrium. Thus each fixed point of the map $\mathscr{T}$
corresponds to a mean field equilibrium.

To show that the map $\mathscr{T}$ has a fixed point, we apply
Kakutani's fixed point theorem. For this, we first identify a compact
set $\mathcal{X} \defeq [0,M]^2\times [\ubarc, \barc]$ of
$\reals_0^3$, with $0 < \ubarc < \barc$ and $M > 0$, and show that
$\mathscr{T}(\mathcal{X}) \subseteq \mathcal{X}$. We include the proof
of this step in Appendix~\ref{ap:compactness-map}. Second, we show that
the map is upper hemicontinuous under the Euclidean topology. To show
this, we need to show that each of the intermediate maps in
Fig.~\ref{fig:map-t} is continuous (or upper hemicontinuous if the map
is a correspondence) under appropriate topologies. We apply Berge's
Maximum Theorem \cite{berge1963topological} to show the map from
$(n_0, n_1,C)$ to $(\kappa, \pi)$ is continuous, and the continuity of
fixed points of continuous contraction mappings \cite{granas2013fixed}
to show the map from $(n_0,n_1, \kappa, C)$ to the value function
$\hat{V}$ is continuous. Finally, we show that the image of
$\mathscr{T}$ is convex and non-empty for all values of
$(n_0, n_1,C)$. Then, it follows by a direct application of Kakutani's
fixed point theorem that the map $\mathscr{T}$ has a fixed point, and
consequently a mean field equilibrium exist in the infinite
system.

\section{Comparative statics}
\label{sec:numerics}

Having shown the existence of a mean field equilibrium for the infinite
system, we now study how the model parameters and the reward function
$f$ affect the equilibrium agent behavior. We perform this
investigation numerically by computing the mean field equilibrium for
a range of parameter values, and studying how the equilibrium
thresholds $(n_0, n_1)$ and the stationary distribution vary. 

As our model is invariant to the proportional scaling of the decision
epoch rates $\lambda$, and the holding rates $\mu_{0,1}$ and
$\mu_{1,0}$, we assume for our computations that $\lambda=1$ and vary
the holding rates. Further, we restrict our attention to the symmetric
case where the holding rates are equal: $\mu_{0,1} = \mu_{1,0} = \mu$.
In all our computation, we set the agent density $\beta = 20$, and the
survival probability $\gamma = 0.95$. For this, we study the mean
field equilibrium in three reward settings: (1) $f(n) = 1/n$ for all
$n$; (2) $f(n) = 1/n^2$ for all $n$; and (3) $f(n) = 1/\sqrt{n}$ for
all $n$.

Before we discuss the results of our numerical investigation, we
briefly describe our approach to compute an (approximate) mean field
equilibrium of the infinite system. 

\subsection{Computation of MFE}

Recall that the mean field equilibria of our model are the fixed
points of the correspondence $\mathscr{T}$. We thus seek to find
(approximate) fixed points of this map. To do this, we adopt a
brute-force approach. We first truncate the state space $\sets$ of the
agent decision problem to $\sets_{200}= \{0,1\} \times \{0,1,\cdots, 199\}$.  We
restrict the thresholds $(n_0,n_1)$ to grid of values in $[0,50]^2$,
where we set the grid resolution $r$ adaptively over different
runs. We do a similar adaptive meshing of the set of values of the
immediate payoff $C$. Having restricted the set of values of
$(n_0,n_1)$ and $C$ thus, we search over all values to find lie in the
image of $\mathscr{T}$, within some pre-specified tolerance
$\epsilon$. We describe this process in detail:
\begin{enumerate}
\item For each value of $(n_0, n_1)$, we perform a binary search on
  $\kappa$ to find a value for which the stationary distribution $\pi$
  of the Markov chain $\mathcal{MC}(n_0, n_1, \kappa)$ restricted to
  $\sets_{200}$ satisfies equation $\eqref{eq:expectation-steady}$
  with a tolerance $\epsilon$. (The stationary distribution $\pi$ is
  obtained by solving the set of linear equations
  \eqref{eq:steady-state}.)
\item For this value of $\kappa$, $(n_0,n_1)$ and each value of $C$,
  we perform value iteration to compute the value function $\hat{V}$,
  again with a tolerance of $\epsilon$, and compute the set of optimal
  thresholds $\mathcal{T}(n_0,n_1,\kappa, C)$. Let $\textsf{dist}(n_0,n_1)$
  denote the distance between $(n_0,n_1)$ and the
  $\mathcal{T}(n_0, n_1, \kappa,C)$ under the Euclidean norm. (Note
    that the latter set is convex and the distance is well defined.)
\item Next, using the stationary distribution $\pi$ and the value
  function $\hat{V}$, we compute the immediate payoff $\tilde{C}$.
\item For each value of $n_0, n_1$ and $C$, we compute
  $d(n_0, n_1, C) \defeq \|C - \tilde{C}\| +
  \textsf{dist}(n_0,n_1)$.
  We output the value of $(n_0,n_1,C)$ that minimizes $d(n_0,n_1,C)$
  over all values.
\end{enumerate}
To make the brute-force search efficient, we run this algorithm
sequentially and adaptively by first identifying candidate regions
where equilibria might exist, and restricting the search to those
regions with lower tolerance and finer grid values.

\subsection{Numerical results}

In Table~\ref{table} we report computationally determined values for three different reward functions, obtained over five different rates for the underlying resource process.

\begin{table}[tb]
\centering
\begin{tabular}{|c|c|c|c|}
\hline
$\mu$ & $f(n) = 1/\sqrt{n}$ & $f(n)=1/n$ & $f(n) = 1/n^2$ \\
\hline
0.1 & 2.80, (1.0, 43.9) & 1.98, (1.0, 4.0)  & 0.14, (1.0, 7.7)\\
\hline
0.5 & 2.39, (4.3, 43.8) & 1.72, (1.0, 4.0) & 0.25, (1.0, 4.7)\\
\hline
1.0 & 2.63 , (1.0, 27.2) & 0.96, (1.0, 10.4) & 0.18, (1.0, 5.0) \\
\hline
10.0 & 2.37, (11.0, 18.2) & 0.93, (4.0, 7.1) & 0.16, (3.0, 5.0) \\
\hline
100.0 & 2.46, (11.0, 12.0) & 0.97, (4.0, 4.0) & 0.80, (1.0, 8.0)\\
\hline
\end{tabular}
\caption{Computationally determined approximate equilibrium values of the payoff for switching $C$,
and the thresholds $n_0$ and $n_1$ used to decide whether to switch or not.  Values are reported in the table as $C, (n_0,n_1)$, for each value for reward function $f$ and the rate $\mu$ at which the resource level changes. \label{table}}
\end{table}

For large values of $\mu$, we see $n_0$ and $n_1$ are close. This is natural because large values of $\mu$ imply that the resource level is changing very quickly, and so the level of resource at the time of the decision has little impact on what the resource level will be at the next time the agent receives a reward.  Thus, the current level of resource $z$ has little impact on the threshold $n_z$. On the other hand, for small values of $\mu$, the thresholds differ significantly with the resource level.

We also observe that, when comparing reward functions $f(n) = 1/\sqrt{n}$, $1/n$, and $1/n^2$, when $f$ decreases more quickly, agents are more willing to switch (the threshold for switching is lower), and the payoff for switching is also lower.

\section{Conclusion}
\label{sec:conclusion}

We have studied a multi-agent location-specific resource-sharing game, and have
established the existence of an equilibrium in this game, and have characterized
each agent's policy in this equilibrium as being a threshold policy.  This
result provides economic insight into such multi-agent resource-sharing games,
and also allows evaluating the effects of designing and modifying these games,
through subsidies or penalties added to natural occurring rewards and costs.

This work also sets the stage for studying information sharing in multi-agent
resource-sharing systems. Our current analysis assumes that agents observe only
the number of other agents and resource level at their current location, and
the locations they have visited in the past.  One may extend this model to
allow for information sharing among the agents as well as between the agents
and a central planner who has access to the current state(s) of the location(s)
that the agent mights switch to. A first question of interest would then be
whether such information sharing necessarily improves social welfare, or
whether it can in fact degrade it.  A second question is how this information
sharing mechanism should be designed to maximize social welfare.


\newpage

\appendixtitleon

\begin{appendices}

\section{Proofs}
\label{ap:few-proofs}

In this appendix we provide proofs of selected results discussed in the main text.
A full proof of the main theorem is technical, and is omitted due to space constraints.
This full proof will appear in a future version of the paper.

\subsection{Structure of optimal strategies}
\smallskip

\begin{proof}[of Lemma \ref{lem:value-decreasing}]
  By Markov property we may assume at $t=0$ the agent makes his
  current decision, and let $t'$ be the time his next decision epoch
  starts. We denote $\nu_{(z,n)}$ as the distribution of
  $(Z_{t'}, N_{t'})$ given $(Z_0, N_0) = (z,n)$, then we can write
    \begin{equation*}
        \vhat(z,n) = \expec_{\nu_{(z,n)}} V(Z_{t'}, N_{t'}).
    \end{equation*}

    Let $\vhat^{(0)}(z,n) = 0$ for all $z$ and $n$. Compute
    $ V^{(m)}(z,n)$ and $\vhat^{(m+1)}(z,n)$ using value iteration for
    the Bellman's equation \eqref{eq:bellman}.
    \begin{align*}
          V^{(m)}(z,n) &=  zf(n) + \gamma \max\{\vhat^{(m)}(z,n), C\}, \\
          \vhat^{(m+1)}(z,n) &=  \expec_{\nu_{(z,n)}}V^{(m)}(Z_{t'}, N_{t'}).    
    \end{align*}
    By convergence of value iteration we have
    $\|V^{(m)}-V\|_\infty \rightarrow 0$ and
    $\|\vhat^{(m)} - \vhat\|_\infty \rightarrow 0$ as
    $m\rightarrow\infty$.  Therefore, it suffices to show
    $\forall m \in\mathbb{N}_0$, for $z\in\{0,1\}$, $V^{(m)}(z,n)$ and
    $\vhat^{(m)}(z,n)$ are non-increasing in $n$.  We can prove this by
    induction on $m$.
 
    The base case follows trivially. Now assume $\vhat^{(m)}(z,n)$ is
    non-increasing in $n$ for $z\in\{0,1\}$, for some
    $m \in \naturals_0$.  From this it is straightforward to conclude
    that $V^{(m)}(z,n) = zf(n) + \gamma\max\{\vhat^{(m)}(z,n), C\}$
    must be non-increasing in $n$ since both $f(n)$ and
    $\vhat^{(m)}(z,n)$ are non-increasing in $n$. Thus, we only need
    to show that $\vhat^{(m+1)}(z,n)$ is also non-increasing in $n$.
    
    Observe showing
    $\vhat^{(m+1)}(z,n) \geq \vhat^{(m+1)}(z,n+1)$ is equivalent to
    showing $\expec_{\nu_{(z,n)}}V^{(m)}(Z_{t'}, N_{t'})$
    $\geq \expec_{\nu_{(z,n+1)}}V^{(m)}(Z_{t'}, N_{t'})$.  We show
    this using a sample path argument by considering two
    processes. Let $(Z_t^1, N_t^1)$ be a copy of
    $\mathcal{MC}(n_0, n_1, \kappa)$ that starts at
    $(Z_0^1, N_0^1) = (z,n)$ and let $(Z_t^2, N_t^2)$ be a copy of
    $\mathcal{MC}(n_0, n_1, \kappa)$ that starts at
    $(Z_0^2, N_0^2) = (z, n+1)$. By carefully coupling the two
    processes, the details of which we omit due to space restrictions,
    it can be shown that for all $t\geq 0$, $Z_t^1 = Z_t^2$ and
    $N_t^1 \leq N_t^2$. Since $(Z^1_{t'}, N^1_{t'}) \sim \nu_{(z,n)}$
    and $(Z^2_{t'}, N^2_{t'}) \sim \nu_{(z, n+1)}$, we have
    \maskforsubmission{ Note that the two processes is a special case
      of Lemma~\ref{lem:compare-different-kappa}, where
      $\kappa_1 =\kappa_2 = \kappa$, and $N_0^2 = N_0^1 + 1$. Thus, we
      could construct the two processes in the same way as in
      Lemma~\ref{lem:compare-different-kappa}, and obtain the result
      that $Z_t^1 = Z_t^2$ and $N_t^1 \leq N_t^2$ in all sample paths
      at all $t \geq 0$. Note that
      $(Z^1_{t'}, N^1_{t'}) \sim \nu_{(z,n)}$ and
      $(Z^2_{t'}, N^2_{t'}) \sim \nu_{(z, n+1)}$, hence we have }
\begin{align*}
  \expec_{\nu_{(z,n)}}V^{(m)}(Z_{t'}, N_{t'}) = \expec [V^{(m)}(Z_t^1, N_t^1)],
\end{align*}
and
\begin{align*}
  \expec_{\nu_{(z,n+1)}}V^{(m)}(Z_{t'}, N_{t'}) = \expec [V^{(m)}(Z_t^2, N_t^2)].
\end{align*}
Since $V^{(m)}(z,n)$ is non-increasing in $n$ for both $z=0,1$,
$V^{(m)}(Z_t^2, N_t^2) \leq V^{(m)}(Z_t^1, N_t^1)$ in all sample
paths, and therefore
$\expec [V^{(m)}(Z_t^1, N_t^1)] \geq \expec[V^{(m)}(Z_t^2, N_t^2)]$,
which completes the proof.
\end{proof}

We now consider the set of optimal thresholds
$\mathcal{T}(n_0, n_1, \kappa, C)$. 
\begin{lemma}
\label{lem:thresholds-convex}
For each $(n_0,n_1)$, $\kappa>0$, and $C>0$, the set of optimal
thresholds $\mathcal{T}(n_0,n_1,\kappa,C)$, as a subset of
$\reals^2_+$, is convex.
\end{lemma}
\begin{proof} From Lemma~\ref{lem:value-decreasing}, we know that the
  value function $\vhat(z,n)$ is non-increasing in $n$ for each
  $z \in \{0,1\}$.  Now for $z \in \{0,1\}$, denote
  $\ubar{n}_z=\max\{n: \vhat(z,n) > C\}$ and
  $\bar{n}_z = \min\{n: \vhat(z,n) < C\}$. For all
  $n \leq \ubar{n}_z$, the optimal action is to stay at the current
  location and for all $n \geq \bar{n}_z$, the optimal action is to
  switch to a different location. For any integer
  $n \in [\ubar{n}_z, \bar{n}_z]$, we have $\vhat(z,n) = C$, meaning
  the agent is indifferent between staying in the current location or
  switching to a different location when the state at the current
  location is $(z,n)$.

  Let $(n^1_0,n^1_1)$ and $(n^2_0,n^2_1)$ be two threshold strategies
  that are both optimal:
  $(n^\ell_0,n^\ell_1) \in \mathcal{T}(n_0,n_1,\kappa,C)$ for
  $\ell \in \{1,2\}$. By the definition of the threshold strategy,
  this implies that $\lfloor n^\ell_z \rfloor \geq \ubar{n}_z$, and
  $\lceil n^\ell_z\rceil \leq \bar{n}_z$ for each $\ell$. This is
  because in the threshold strategy $(n^\ell_0,n^\ell_1)$, at state
  $(z,n)$, the agent stays at her current location for all
  $n < \lfloor n^\ell_z\rfloor$, switches to a different location for
  all $n > n^\ell_z$ and and stays with probability
  $n^\ell_z - \lfloor n^\ell_z \rfloor$ and switches with the
  remaining probability if the number of agents $n$ is equal to
  $\lfloor n^\ell_z \rfloor$. Since this is true for each $\ell$, we
  have for any $\alpha \in (0,1)$,
  $\lfloor \alpha n^1_z + (1-\alpha) n^2_z \rfloor \geq \ubar{n}_z$,
  and $\lceil \alpha n^1_z + (1-\alpha) n^2_z \rceil \leq \bar{n}_z$.
  This implies that a threshold strategy that at state $(z,n)$ stays
  in the current location if
  $n < \lfloor \alpha n^1_z + (1-\alpha) n^2_z \rfloor$, switches to a
  different location if
  $n > \lfloor \alpha n^1_z + (1-\alpha) n^2_z \rfloor$, and stays
  with probability
  $(n+1-\lfloor \alpha n^1_z + (1-\alpha) n^2_z \rfloor)$ and switches
  otherwise if $n = \lfloor \alpha n^1_z + (1-\alpha) n^2_z \rfloor$
  is also optimal. This implies that
  $(\alpha n^1_0 + (1-\alpha)n^2_0, \alpha n^1_1 + (1-\alpha)n^2_1)$
  also lies in the set $\mathcal{T}(n_0,n_1,\kappa,C)$, and hence the
  latter set is convex.
\end{proof}

\subsection{Restriction of $\mathscr{T}$ to a compact set}
\label{ap:compactness-map}
In this section, we show that the map $\mathscr{T}$ can be restricted
to a compact subset $\mathcal{X}$ of $\reals^3_+$, whose image
$\mathscr{T}(\mathcal{X})$ is a subset of $\mathcal{X}$.

Towards that goal, we define $\barc$, $\ubarc$ as follows:
\begin{align*}
  \barc &\defeq \frac{f(1)}{1-\gamma}, \\
  \ubarc &\defeq \frac{\lambda}{\lambda + \mu_{1,0} + \beta\lambda}\frac{\mu_{0,1}}{\lambda + \mu_{0,1} + \beta\lambda} e^{-\frac{\beta}{1-\gamma}} f(1).
\end{align*}
It is straightforward that $0 < \ubar C < f(1) < \bar C$. Also, for
$n \in \naturals_0$, define
\begin{align}
  \label{eq:reward-upper-bound}
  g(n) &\defeq \frac{1}{1-\gamma} \left[f(\frac{\sqrt{n}}{2}) + \exp(-\frac{\sqrt{n}}{8}) + \frac{2}{\sqrt{\log n}}\right]\notag\\
       &\quad + \frac{\gamma^{\lfloor (\log n)^{1/2} \rfloor}}{1-\gamma}f(1).
\end{align} 
Note $g(n)$ is decreasing in $n$ and $g(n) \rightarrow 0$ as
$n \rightarrow +\infty$. We pick $M$ such that
\begin{equation*}
  M \defeq \min\{n: g(n) < (1-\gamma)\ubarc\}.
\end{equation*}

Define $\mathcal{X} = [0,M]^2 \times[\ubarc, \barc]$. The following
theorem shows that the image of map $\mathscr{T}$ is a subset of
$\mathcal{X}$ for all values therein.
\begin{theorem}
    \label{thm:Compactness}
    For all $ (n_0, n_1, C)\in \mathcal{X}$,
    $\mathscr{T}(n_0, n_1, C) \subseteq \mathcal{X}$.
\end{theorem}

We prove this theorem in two steps. First, in
Lemma~\ref{lem:compactness-of-c}, we show that for all values of
$(n_0, n_1, C)$ in $\mathcal{X}$, the value of $\tildec$ lies in
$[\ubarc,\barc]$. Then, we show in
Lemma~\ref{lem:compactness-threshold}, that for all values of
$(n_0, n_1, C)$ in $\mathcal{X}$, the optimal thresholds must always
be less than $M$. We begin with the first lemma.
\begin{lemma}
    \label{lem:compactness-of-c}
    For any $(n_0, n_1, C)\in [0, +\infty)^2\times[\ubarc, \barc]$ and
    $\kappa\in [\beta\lambda(1-\gamma), \beta\lambda]$, let $\vhat$ be
    the solution of the Bellman equation \eqref{eq:bellman} with given
    parameters $(n_0, n_1, C, \kappa)$. Let $\pi$ be the unique
    stationary distribution of
    $\overline{\mathcal{MC}}(n_0, n_1, \kappa)$. Let
    $\tildec \defeq \sum_{(z,n)\in\sets}\pi(z,n) \vhat(z,n+1)$. Then
    $\tildec \in [\ubarc, \barc]$.
\end{lemma}

\begin{proof}
    We first show $\tildec \leq \barc$. We have 
    \begin{align}
      \label{eq:c-less-than-vhat}
      \tildec &= \sum_{(z,n)\in\sets} \pi(z,n)\vhat(z,n+1)\notag \\
              &\leq \max_{(z,n)\in\sets} \vhat(z,n+1)\notag\\
              &= \max_{z\in\{0,1\}} \vhat(z,1),
    \end{align}
    where the last equality is implied by Lemma~\ref{lem:value-decreasing}. Also note
    $\sP_{exit}(z, 1) = \sP_{sur}(z, 1)=0$ for $z\in\{0,1\}$, therefore
    \begin{align}
        \label{eq:vhat-n-equals-one}
      \vhat(z,1) &= \sP_{dec}(z,1)V(z,1) + \sP_{res}(z,1)\vhat(1-z,1)\notag\\
                 &\quad + \sP_{arr}(z,1)\vhat(z,2), \quad \text{for $z\in\{0,1\}$}. 
    \end{align}
    From Lemma~\ref{lem:value-decreasing} we have
    $\vhat(z,2) \leq \vhat(z,1)$ for $z\in\{0,1\}$. Further, assume
    $z^*\in\{0,1\}$ attains $\max_{z \in \{0,1\}} \vhat(z,1)$, then
    $\vhat(1-z^*,1) \leq \vhat(z^*, 1)$, and
    \eqref{eq:vhat-n-equals-one} becomes
    \begin{align}
        \label{eq:vhat-less-than-v}
      \vhat(z^*, 1) &\leq \sP_{dec}(z^*,1)V(z^*,1) + \sP_{res}(z^*,1)\vhat(z^*,1)\notag\\
   &\quad + \sP_{arr}(z^*, 1)\vhat(z^*, 1).
    \end{align}
    We have
    $1 - \sP_{res}(z^*,1) - \sP_{arr}(z^*,1) = \sP_{dec}(z^*,1) =
    \lambda/(\lambda + \mu_{z^*, 1-z^*} + \kappa) > 0$,
    along with \eqref{eq:vhat-less-than-v} this gives
    $\vhat(z^*, 1) \leq V(z^*,1)$.  Thus, we have
    \begin{align}
        \label{eq:vhatlessthanv}
      \vhat(z^*, 1) &\leq V(z^*, 1)\notag\\
                    &= z^*f(1) + \gamma \max\{C, \vhat(z^*, 1)\} \notag \\
                    &\leq z^*f(1) + \gamma \max\{C, \vhat(z^*, 1)\}. 
    \end{align}

    If $\vhat(z^*, 1) < C$, then by \eqref{eq:c-less-than-vhat} we
    have
    \begin{equation*}
        \tildec \leq \max_{z\in\{0,1\}} \vhat(z,1) = \vhat(z^*, 1) < C \leq \barc.
    \end{equation*}
    Otherwise $C \leq \vhat(z^*, 1)$ and \eqref{eq:vhatlessthanv}
    becomes
    \begin{align*}
         \vhat(z^*, 1) \leq z^* f(1) + \gamma\vhat(z^*, 1),  
     \end{align*}
     which gives us
    \begin{equation*}
        \vhat(z^*, 1) \leq  \frac{z^*f(1)}{1-\gamma} \leq \frac{f(1)}{1-\gamma} = \barc.
    \end{equation*}
    Hence in both cases we have $\tildec \leq \barc$.  

    Next we show $\tildec \geq \ubarc$. We have
    \begin{equation}
        \label{eq:lower-tildec}
        \tildec = \sum_{(z,n)\in\sets} \pi(z,n)\vhat(z,n+1) \geq \sum_{z=0,1} \pi(z,0)\vhat(z,1).
    \end{equation}  
    Note that
    \begin{align*}
            \vhat(1, 1) &= \sP_{dec}(1,1) V(1,1) + \sP_{res}(1,1) \vhat(0,1)\\
      &\quad +\sP_{arr}(1,1) \vhat(1,2) \\ 
            &\geq \sP_{dec}(1,1) V(1,1) \\
            & = \sP_{dec}(1,1)[f(1) + \gamma\max\{\vhat(1,1), C\}] \\
            &\geq \sP_{dec}(1,1) f(1),
    \end{align*}
    and
    \begin{align*}
      \vhat(0,1) &= \sP_{dec}(0,1) V(0,1) + \sP_{res}(0,1) \vhat(1,1)\\
                 &\quad +\sP_{arr}(0,1) \vhat(0,2) \\
                 &\geq \sP_{res}(0,1)\vhat(1,1) \\
                 &\geq \sP_{res}(0,1)\sP_{dec}(1,1) f(1).
    \end{align*}
    Therefore, with \eqref{eq:lower-tildec} we have
    \begin{equation*}
        \tildec \geq \sP_{res}(0,1)\sP_{dec}(1,1) f(1) \sum_{z=0,1}\pi(z,0).
    \end{equation*}

    Let $N^1_t$ be the number of agents in our system at time $t$, and
    $N^2_t$ be that of an $M/M/\infty$ queue with arrival rate
    $\kappa$ and service rate $\lambda(1-\gamma)$. Assume
    $N_0^1 = N_0^2$. Using a coupling argument, which we omit due to
    space limitations, it can be shown that $N^2_t$ stochastically
    dominates $N^1_t$ for all $t \geq 0$. Therefore,
    \maskforsubmission{We have in
      Lemma~\ref{lem:three-systems-sto-dominance} shown $N^2_t$
      stochastically dominates $N^1_t$ for all $t \geq 0$, therefore}
    \begin{align*}
        \sum_{z=0,1} \pi(z,0) &= \sum_{z=0,1} \lim_{t \rightarrow +\infty} \prob(N^1_t = 0, Z_t=z) \\
        &= \lim_{t \rightarrow +\infty} \prob(N^1_t = 0) \\
        &\geq \lim_{t \rightarrow +\infty} \prob(N^2_t = 0) \\
        &= e^{-\frac{\kappa}{\lambda(1-\gamma)}},
    \end{align*}
    where the last equality follows from the steady state distribution
    of the $M/M/\infty$ queue.

    Therefore we have
    \begin{align*}
      \tildec &\geq \sP_{dec}(1,1)\sP_{res}(0,1) f(1) \sum_{z=0,1} \pi(z,0) \\
              &= \left(\frac{\lambda}{\lambda + \mu_{1,0} + \kappa}\right) \left(\frac{\mu_{0,1}}{\lambda + \mu_{0,1} + \kappa}\right) e^{-\frac{\kappa}{\lambda(1-\gamma)}} f(1) \\
              & \geq \left(\frac{\lambda}{\lambda + \mu_{1,0} + \beta\lambda}\right)\left( \frac{\mu_{0,1}}{\lambda + \mu_{0,1} + \beta\lambda}\right) e^{-\frac{\beta}{1-\gamma}} f(1) = \ubar{C},
    \end{align*}
    where the last inequality is achieved on picking
    $\kappa = \beta\lambda$.  
  \end{proof}

Next we restrict our choice of thresholds $(n_0, n_1)$ to a compact set. 

\begin{lemma}
    \label{lem:compactness-threshold}    
    Given $(n_0, n_1) \in [0, +\infty)^2$, $C \in [\ubarc, \barc]$ and
    $\kappa\in[ \beta\lambda(1-\gamma), \beta\lambda]$, for any $(\tilden_0,
    \tilden_1) \in \mathcal{T}(n_0, n_1, \kappa, C)$, we have $\tilden_0 \leq
    M$ and $\tilden_1 \leq M$.
\end{lemma} 

\begin{proof}
    
  Suppose at $t=0$ the $0^{th}$ decision epoch the agent starts, and
  there are $n$ agents at the node. We show for large enough $n$, the
  total expected payoff the agent receives if he chooses to stay at
  $t=0$ will be less than what he would receive if he chooses to
  switch to a different location.

  Let $\tau \in \naturals_0$ be the first decision epoch the agent
  chooses to leave. We seek to show that $\tau = 0$. Suppose, for the
  sake of arriving at a contradiction, $\tau \geq 1$. Let $R_i$ be the
  immediate reward the agent receives at his $i^{\text{th}}$ decision
  epoch if he is still at the current node at the time he makes his
  $i^{\text{th}}$ decision, and $R$ be the total reward he receives on
  choosing to stay at $t=0$. We have:
  \begin{align}
    \label{eq:total-reward}
    \expec[R] &= \sum_{i=0}^{+\infty} \gamma^i \expec[R_i\one{\{\tau > i\}}] + \expec[\gamma^\tau C] \notag \\
              &\leq \sum_{i=0}^{+\infty} \gamma^i \expec[R_i\one{\{\tau > i\}}] + \gamma C.
  \end{align}
  The first term in the first equation of \eqref{eq:total-reward} is
  the expected total reward until the agent chooses to leave, and the
  second term is the aggregated expected payoff on leaving. The
  inequality follows from the assumption that $\tau \geq 1$. Our goal
  is to show $\expec[R_i\one{\{\tau > i\}}]$ vanishes as
  $n \rightarrow +\infty$, hence $\expec[R] \rightarrow \gamma C < C$
  as $n \rightarrow +\infty$, and since $C$ is the aggregated expected
  payoff on leaving at $t=0$, we would have $\tau = 0$.
    
  Let $T_i$ be the time the $i^{th}$ decision epoch of the agent
  starts. We have $T_i \sim \text{Gamma}(i, \lambda)$ since the
  interval between two consecutive decision epochs are
  i.i.d. $\exp(\lambda)$ and $T_i$ is the sum of $i$ such intervals.
    
  Let $n' = \frac{1}{2}\sqrt{n}$. We have for all
  $ i \in \mathbb{N}^+$,
    \begin{align}
        \label{eq:upper-bound-expec-reward-per-epoch}
            \expec[R_i\one{\{\tau \geq i\}}] & = \expec[R_i\one{\{\tau\geq i\}} | N_{T_i} \geq n']\prob(N_{T_i} \geq n')\notag \\ &\quad + \expec[R_i\one{\{\tau \geq i \}} | N_{T_i} < n']\prob(N_{T_i} < n') \notag \\
            & \leq f(n')\prob(N_{T_i} \geq n') + \prob(N_{T_i} < n')f(1) \notag \\
            & \leq f(n') + \prob(N_{T_i} < n')f(1).
    \end{align}
    We first show $\prob(N_{T_i} < n')$ vanishes as
    $n \rightarrow +\infty$. Consider an alternative system with $n$
    agents at $t=0$ where each agent stays an $\exp(\lambda)$ time and
    then leaves the system. Let $N'_t$ be the number of agents in the
    system at time $t$. For any agent in this alternative system, the
    probability he is still in the system at time $T_i$ is
    $e^{-\lambda T_i}$, therefore we have
    $N'_{T_i} \sim \text{Bin}(n, e^{-\lambda T_i})$.
    
    Using a similar argument as in the proof of
    Lemma~\ref{lem:value-decreasing},\maskforsubmission{ and
      Lemma~\ref{lem:ergodicity},} we can show for all $t\geq 0$,
    $N_t$ is no less than $N'_t$ in all sample paths, i.e., $N_t$
    stochastically dominates $N'_t$ in the first order.  Let
    $k = \lfloor (\log n)^{\frac{1}{2}} \rfloor$. Note
    $N'_{T_i} \geq N'_{T_k}$, for all $ i \leq k$. Thus, pick
    $t_{k} = \frac{1}{2\lambda}\log n$. We have
    \begin{align*}
        \prob(N_{T_i} < n') & \leq \prob(N'_{T_i} < n') \\
                      & \leq \prob(N'_{T_k} <  n') \\
                      & \leq \prob(N'_{T_k} < n' | T_k < t_k ) \prob(T_k < t_k) + P(T_k \geq t_k) \\
                      & \leq \prob(N'_{t_k} < n' | T_k < t_k) \prob(T_k < t_k) + \prob(T_k \geq t_k ) \\
                      & \leq \prob(N'_{t_k} < n') + \prob(T_k \geq t_k).
    \end{align*}

    We have $N'_{t_k} \sim \text{Bin}(n, e^{-\lambda t_k})$. Given $ne^{-\lambda
        t_k} = \sqrt{n} > \frac{1}{2}\sqrt{n} = n'$, we can apply the Chernoff
    bound to obtain,
    \begin{align*}
        \prob(N'_{t_k} < n') & \leq \exp(-\frac{(ne^{-\lambda t_k} - n')^2}{2ne^{-\lambda t_k}}) \\
                              & =  \exp(-\frac{n'}{4}) \\
                              & = \exp(-\frac{\sqrt{n}}{8}).
    \end{align*} 
    Also, by Markov's inequality we have
    \begin{align*}
        \prob\left(T_k \geq t_k\right)  \leq \frac{\expec [T_k]}{t_k}  = \frac{k}{\lambda t_k}  = \frac{2\lfloor (\log n)^{\frac{1}{2}} \rfloor}{\log n} 
         \leq \frac{2}{\sqrt{\log n}}.
    \end{align*}
    Therefore,
    \begin{equation*}
        \prob(N_{T_i} < n') \leq \exp(-\frac{\sqrt{n}}{8}) + \frac{2}{\sqrt{\log n}}.
    \end{equation*}
    By \eqref{eq:upper-bound-expec-reward-per-epoch} and using the fact that
    $n'=\frac{1}{2}\sqrt{n}$, we have
    \begin{equation*}
        \expec[R_i\one{\{\tau > i\}}] \leq f(\frac{\sqrt{n}}{2}) + \exp(-\frac{\sqrt{n}}{8}) + \frac{2}{\sqrt{\log n}},
    \end{equation*}
    for all $i = 1, \ldots, k$. Therefore, we have
    \begin{align*}
        \sum_{i=0}^{\infty} \gamma^i \expec[R_i \one{\{\tau > i\}}] & =  \sum_{i=0}^{k - 1} \gamma^i \expec[R_i \one{\{\tau > i\}}] +  \sum_{i=k}^{\infty} \gamma^i \expec[R_i \one{\{\tau > i\}}] \\
        & \leq \sum_{i=0}^{k-1} \gamma^i \left[ f(\frac{\sqrt{n}}{2}) + \exp(-\frac{\sqrt{n}}{8}) + \frac{2}{\sqrt{\log n}} \right]+  \sum_{i=k}^{\infty} \gamma^i f(1) \\
        & \leq \frac{1}{1-\gamma} \left[ f(\frac{\sqrt{n}}{2}) + \exp(-\frac{\sqrt{n}}{8}) + \frac{2}{\sqrt{\log n}} \right] + \frac{\gamma^k}{1-\gamma}f(1).
    \end{align*}

    The righthand side is denoted as $g(n)$ in \eqref{eq:reward-upper-bound}.
    As $n\rightarrow +\infty$, $g(n)$ goes to 0. Hence by picking $M$ such that 
    \begin{equation*}
      M \defeq \min\{n: \frac{1}{1-\gamma}g(n) + \frac{\gamma^k}{1-\gamma}f(1) < (1-\gamma)\ubarc\},
    \end{equation*}
    we have
    \begin{equation*}
        \expec[R] < (1-\gamma)\ubarc + \gamma C < C,
    \end{equation*}
    for all $n \geq M$.  Thus we have proved choosing to stay at $t=0$
    when there are more than $M$ agents at the location is suboptimal
    no matter what the resource level is. This contradicts the
    assumption that $\tau \geq 1$.
\end{proof}
\end{appendices}

\end{document}